\documentclass[graybox]{svmult}
\usepackage{natbib}

\usepackage{mathtools}

\usepackage{type1cm}        
\usepackage{makeidx}         
\usepackage{graphicx}        
\usepackage{multicol}        
\usepackage[bottom]{footmisc}

\usepackage{newtxtext}       %
\usepackage{newtxmath}       

\makeindex             
                       
\usepackage{amssymb}
\usepackage{amsmath}
\usepackage{multirow}
\usepackage{subfig}
\usepackage{graphicx}
\graphicspath{ {images/} }

\newtheorem{remark}{Remark}
\newcounter{model_count}  
\begin{document}
\title*{Mixed Integer Programming for Searching Maximum  Quasi-Bicliques}
%
%
\author{Dmitry I. Ignatov 
\and 
Polina Ivanova \and  Albina Zamaletdinova  
}
\authorrunning{Ignatov et al.} 
%
\tocauthor{Dmitry I. Ignatov}
\institute{
Dmitry I. Ignatov \at National Research University Higher School of Economics, Moscow and St. Petersburg Department of Steklov Mathematical Institute of Russian Academy of Sciences, Russia, \email{dignatov@hse.ru} (0000-0002-6584-8534) 
\and
Polina Ivanova  \at National Research University Higher School of Economics, Moscow, Russia, \email{ivanova.p.m@gmail.com} (0000-0001-6010-7991) 
\and
Albina Zamaletdinova \at National Research University Higher School of Economics, Moscow, Russia, \email{aazamaletdinova\_1@edu.hse.ru}
}

\maketitle 

\abstract*{
This paper is related to the problem of finding the maximal quasi-bicliques in a bipartite graph (bigraph). A quasi-biclique in the bigraph is its ``almost'' complete subgraph. The relaxation of completeness can be understood variously; here, we assume that the subgraph is a $\gamma$-quasi-biclique if it lacks a certain number of edges to form a biclique such that  its density is at least $\gamma \in (0,1]$. For a bigraph and fixed $\gamma$, the problem of searching for the maximal quasi-biclique consists of finding a subset of vertices of the bigraph such that the induced  subgraph is a quasi-biclique and its size is maximal for a given graph.
Several models based on Mixed Integer Programming (MIP) to search for a quasi-biclique are proposed and tested for working efficiency.
An alternative model inspired by biclustering is formulated and tested; this model simultaneously maximizes both the size of the quasi-biclique and its density, using the least-square criterion similar to the one exploited by  triclustering \textsc{TriBox}.
}

\abstract{
This paper is related to the problem of finding the maximal quasi-bicliques in a bipartite graph (bigraph). A quasi-biclique in the bigraph is its ``almost'' complete subgraph. The relaxation of completeness can be understood variously; here, we assume that the subgraph is a $\gamma$-quasi-biclique if it lacks a certain number of edges to form a biclique such that  its density is at least $\gamma \in (0,1]$. For a bigraph and fixed $\gamma$, the problem of searching for the maximal quasi-biclique consists of finding a subset of vertices of the bigraph such that the induced  subgraph is a quasi-biclique and its size is maximal for a given graph.
Several models based on Mixed Integer Programming (MIP) to search for a quasi-biclique are proposed and tested for working efficiency.
An alternative model inspired by biclustering is formulated and tested; this model simultaneously maximizes both the size of the quasi-biclique and its density, using the least-square criterion similar to the one exploited by  triclustering \textsc{TriBox}.
}

\section{Introduction}\label{sec:intro}
 There are many data sources that can be represented as a bipartite graph; for example, in recommender systems and web stores, users can interact with different items like movies, books, clothes, and other products. The most commonly studied data usually has a structure of a bipartite graph whose vertices form two disjoint sets. For example, social network data, where a binary relation between two sets show interactions between people and communities, advertisement data with a set of consumers and a corresponding set of products and so on.

In this study, we are interested in the analysis of such bipartite data and search for dense communities, where almost all elements are connected. A situation where all elements of a community are involved can be described by a concept of a biclique or a complete subgraph of a bipartite graph.

Unfortunately, the community completeness requirement excludes almost complete communities frequently met in real-world data. Due to this reason, we allow some edges to be absent and introduce the concept of quasi-biclique. In order to bound the size of quasi-biclique, we can use the subgraph minimal density or the maximum number of absent edges needed to complete a subgraph.

The problem of searching for maximal quasi-clique is NP-hard~\citep{8_on_maximum_quasi_clique} as well as the problem of searching for maximal quasi-biclique~\citep{1_quasi_complexity}; the maximum edge biclique problem is known to be NP-complete~\citep{Peeters:2003}. Many algorithms that solve those problems are being developed~\citep{Wang:2013,5_massive_quasi_cliques,3_financial_ratios,6_protein_protein}. For instance,  \cite{exact_mip_veremryev} offered an exact Mixed Integer Programming model for searching for maximal quasi-clique but the case of bipartite graphs for quasi-bicliques was not yet considered within the MIP framework.

The aim of this paper is to propose a Mixed Integer Programming models for finding a maximum quasi-bicliques in a bipartite graph and compare the results obtained by those models with those of existing algorithms.

The paper is organised as follows. Section~\ref{sec:relwork} introduces several basic definitions, namely biclique, quasi-bicliques, and its density and provides a short overview of related work along with important propositions on algorithmic aspects. Section~\ref{sec:models} proposes two Mixed Integer Programming models for quasi-biclique search. In Section~\ref{sec:data}, the chosen datasets are described. Section~\ref{sec:experiments} summarises the experimental results. Section~\ref{sec:results_concl} concludes the paper.

\section{Maximum quasi-cliques and quasi-bicliques}\label{sec:relwork}

\subsection{Basic definitions}\label{ssec:def}
Let us introduce several basic notions.
\begin{definition} In a graph $G=(V, E)$ a subgraph $G'=(V', E')$, where $V' \subseteq V$, $E' \subseteq E$, is called a vertex-induced subgraph. Let us denote such graph as $G[V']$.
\end{definition} 

\begin{definition} A complete subgraph of a graph $(V, E)$ is called a clique.
\end{definition}
\begin{definition} A complete bipartite subgraph in a bipartite graph $(U, V, E \subseteq U \times V)$ is called a biclique.
\end{definition}

\begin{definition} The density of an arbitrary graph is the ratio of the number of edges to the maximum possible number of edges.
\end{definition}

The density of a bipartite graph $G=(V,U,E)$ is $\rho = \frac{|E|}{|V||U|}$.

\begin{definition} A subgraph $G'=(V', E')$ of a given graph $G=(V,E)$ is called $f(k)$-dense, if $G'$ is a subgraph induced by a vertex subset $V' \subseteq V$, $|V'| = k$ and $|E'| \geq f(k)$, where $f: Z_{+} \rightarrow R_+$ is a chosen function.
\end{definition}

\begin{definition} A $\gamma$-quasi-biclique in a bipartite graph $G=(U, V, E)$ is its bipartite induced subgraph $G'=(V', U', E'\subseteq V' \times U')$ with the density at least $\gamma \in (0,1]$.
\end{definition}

\subsection{Maximum quasi-cliques}
\par Let us consider properties and searching algorithms of cliques in a graph $G=(V,E)$.
\par For a graph $G=(V,E)$ and a fixed $\gamma \in (0,1]$ we need to find a $V'\subseteq V$ such that $G[V']$ is a $\gamma$-quasi-clique and $|V'|$ is maximal.
\par Problem of searching for maximum quasi-clique as well as the problem of searching for maximum clique is NP-hard \citep{1_quasi_complexity},\citep{8_on_maximum_quasi_clique}. In addition to that, the assumption of graph incompleteness leads to the loss of useful properties of a clique. For instance, inheritance property which is used in most maximum clique searching algorithms does not hold. Namely, if $G[V]$ is a clique, then $G[V']$ is a clique as well, where $V'$ is a subset of $V$. This property does not hold for $\gamma$-quasi-cliques: i.e. a subset of a $\gamma$-quasi-clique is not necessarily a $\gamma$-quasi-clique.
\par However, for quasi-cliques we can define the property of quasi-inheritance: $\gamma$-quasi-clique with $|V|>1$ is a strict superset to a $\gamma$-quasi-clique with $|V|-1$ vertices \citep{8_on_maximum_quasi_clique}. 

\subsection{Maximum quasi-bicliques}
\par The problem of maximum quasi-biclique in a bipartite graph $G=(U,~V,~E)$ with fixed $\gamma \in (0,1]$ is to find $U' \subseteq U$ and $V' \subseteq V$ such that vertex-induced subgraph $G[U',V']$ is a $\gamma$-quasi-biclique of size $|U'| + |V'|$, maximum for this graph. Let us denote a maximum $\gamma$-quasi-biclique in the graph $G$ by $\omega_{\gamma}\left(G\right)$
\par Let us consider several commonly met definitions of quasi-biclique. In \cite{1_quasi_complexity}, we can find the following definition.
\begin{definition}\label{def:delta_bicluque} 
A induced subgraph $G'[U', V']$ is called a $\delta$-quasi-biclique ($0\leq \delta\leq 0.5$) in a bipartite graph $G=(U, V, E)$ if:
\begin{enumerate}
\item $\forall u \in U': d\left(u, V'\right) = \Large{\vert}\{ v \in V' \vert (u, v) \in E \}\Large{\vert}\geq \left(1-\delta\right) \cdot \vert V' \vert$,
\item $\forall v \in V': d\left(v, U'\right) = \Large{\vert}\{ u \in U' \vert (u, v) \in E \}\Large{\vert} \geq \left(1-\delta\right) \cdot \vert U' \vert$ .
\end{enumerate} 
\end{definition}

\par In order to consider the third definition of quasi-biclique~\cite{3_financial_ratios}, let us introduce some useful notations. The neighbourhood of a vertex $v\in V$ in a graph $G=(V,E)$ is a set of vertices $\Gamma(v)=\{u\in V  | (u,v) \in E \}.$

\par For a vertex set $V' \subseteq V$ and a vertex $v\in V\setminus V'$ let us denote a set of vertices from $V'$ adjacent to $v$ as $\Gamma_{V'}(v)=\{u|(u,v)\in E\&u \in V'\}$. By a set of vertices $\Gamma(V')={\cup_{v\in V'}\Gamma(v)}$ we denote a loose neighbourhood of subset $V'$

\begin{definition} \label{def:epsilon_bicluque}
A subgraph $G'[U', V']$ of a bipartite graph $G(U, V, E)$ is called an $\epsilon$-quasi-biclique, if for some small positive integer $\epsilon$:
\begin{enumerate}
\item $\forall u \in U'  \vert V' \vert - \vert \Gamma_{V'}(u) \vert \leq \varepsilon$,
\item $\forall v \in V'  \vert U' \vert - \vert \Gamma_{U'}(u) \vert \leq \varepsilon$.
\end{enumerate} 
\end{definition}

\begin{remark}
Obviously, Definitions \ref{def:delta_bicluque} and \ref{def:epsilon_bicluque} of quasi-bicliques can be reduced to the definition of $\gamma$-quasi-biclique.
\begin{enumerate}
\item In Definition~\ref{def:delta_bicluque}, let us sum the first condition over all vertices from $U'$. We get, that $\displaystyle \sum_{u \in U'} d\left(u, V'\right) \geq \left(1-\delta\right) \cdot \vert V' \vert \vert U' \vert$, where $\displaystyle\sum_{u \in U'} d\left(u, V'\right)$ is a number of edges in a $\delta$-quasi-biclique, $\vert V' \vert \vert U' \vert$ is the maximum possible number of edges in a bipartite graph. Thus a $\delta$-quasi-biclique is a $\gamma$-quasi-biclique with $\gamma = 1 -\delta$. Both definitions of quasi-biclique are equivalent if $\gamma \in [0.5, 1]$.

\item By summing both conditions over sets $U'$ and $V'$, respectively, in Definition \ref{def:epsilon_bicluque} we get:
\begin{equation*}
\dfrac{\sum_{u \in U'}\Gamma_{V'}(u)}{\vert U' \vert \vert V' \vert} \geq 1 - \dfrac{\varepsilon}{\vert V' \vert}, \ \dfrac{\sum_{v \in V'}\Gamma_{U'}(v)}{\vert U' \vert \vert V' \vert} \geq 1 - \dfrac{\varepsilon}{\vert U' \vert}. 
\end{equation*} 
Since $\displaystyle \sum_{u \in U'}\Gamma_{V'}(u) = \sum_{v \in V'}\Gamma_{U'}(v)$ is a number of edges in an $\varepsilon$-quasi-biclique $G[U',V']$, then the density of $G[U',V']$ is:
\begin{equation*}
\rho(G[U',V']) \geq 1 - \dfrac{\varepsilon }{\min(\vert U' \vert, \vert V' \vert)}.
\end{equation*}
Bounding the size of a quasi-clique vertex sets from below $\omega_l^{(1)} \leq \vert U' \vert $ and $\omega_l^{(2)} \leq \vert V' \vert $, we can establish  a connection between these definitions. If we let 
\begin{equation*}
 \gamma= 1 -  \dfrac{\varepsilon}{\min(\omega_l^{(1)}, \omega_l^{(2)})},
\end{equation*}
we obtain that $G[U',V']$ is a  $\gamma$-quasi-clique under condition $\varepsilon \in [0,\min(\omega_l^{(1)}, \omega_l^{(2)}))$.
\end{enumerate}
\end{remark}
\par Most properties of quasi-cliques naturally fulfils for quasi-bicliques as well. However, since the density definition of quasi-biclique differs from the case of quasi-clique and the maximum number of edges is a function of two variables with no convex properties, most algorithms searching for maximum quasi-clique are not directly applicable to search for maximum quasi-biclique.
\par \cite{8_on_maximum_quasi_clique} established inequalities for upper bounds for the size of maximum quasi-clique shown below.

\begin{proposition}\label{prop:quasi_clique_up}
In a graph $G=(V, E)$ with $\vert V \vert = n$ and $\vert E \vert = m$ the maximum size of a quasi-clique $\omega_{\gamma}(G)$ satisfies the following inequality:
\begin{equation}
\omega_{\gamma}(G) \leq \dfrac{\gamma + \sqrt{\gamma + 8\gamma m}}{2\gamma}.
\end{equation}
\end{proposition}
\par In order to obtained similar bound for a quasi-biclique, we need to allow the following conditions on quasi-biclique.

\begin{proposition}\label{prop:quasi_biclique_up}
In a bipartite graph $G(U, V, E)$, with $\vert U \vert = n_U$, $\vert V \vert = n_V$ and $\vert E \vert = m$, the maximum size of a quasi-biclique $\omega_{\gamma}(G)$ satisfies the following inequalities:
\begin{enumerate}
\item $\omega_{\gamma}(G) \leq \sqrt{\dfrac{4m}{\gamma}}$, for balanced quasi-biclique (the sizes of two vertex sets $U$ and $V$ are equal),
\item $\omega_{\gamma}(G) \leq \min \left\lbrace (2 + \theta)\cdot\sqrt{\dfrac{m}{\gamma (1 - \theta)}}, \ \left(1 + \dfrac{1}{1 - \theta}\right)\cdot \sqrt{\dfrac{m (1 + \theta)}{\gamma}} \right\rbrace$, if $\theta \in (0,1)$ and sizes of vertex sets differ from each other by no more than in $\theta$.
\end{enumerate}
\end{proposition}

\begin{proof}Let $U'$ and $V'$ be vertex sets of a maximum $\gamma$-quasi-biclique and let $n_{U'}$ and $n_{V'}$ be their cardinalities, respectively. 
\begin{enumerate}
\item For balanced quasi-clique $n_{U'} = n_{V'}$, hence $\omega_{\gamma}(G) = 2\cdot n_{U'}$. Obviously, that the maximum possible number of edges in a quasi-biclique in less than the total number of graph edges. Then
\begin{equation*}
\gamma \cdot n_{U'}^2 = \gamma \cdot \left(\dfrac{\omega_{\gamma}(G)}{2}\right)^2 \leq m, 
\end{equation*}

\item $\gamma$-quasi-biclique is ``almost'' balanced when $(1-\theta)~n_{V'}~\leq n_{U'}~\leq(1+\theta)~n_{V'}$. Thus,
\begin{gather*}
\omega_{\gamma}(G) = n_{U'} + n_{V'} \leq (2 + \theta) n_{V'} \Rightarrow  \\
m \geq \gamma \cdot n_{U'} \cdot n_{V'} \geq \gamma \cdot (1 - \theta) \cdot n_{V'}^2 \geq \gamma \cdot (1 - \theta) \cdot \left(\dfrac{\omega_{\gamma}(G)}{2 + \theta} \right)^2  \Rightarrow \\
\Rightarrow \omega_{\gamma}(G) \leq \sqrt{\dfrac{m (2 + \theta)^2}{\gamma (1 - \theta)}}.
\end{gather*}

Analogously, 
\begin{gather*}
\omega_{\gamma}(G) \leq \left(1 + \frac{1}{1 - \theta}\right) n_{U'},\  \gamma \cdot n_{U'} \cdot n_{V'} \geq \gamma \cdot \frac{1}{(1 + \theta)} n_{U'}^2  \Rightarrow \\
\Rightarrow \omega_{\gamma}(G) \leq \left(1 + \dfrac{1}{1 - \theta} \right) \sqrt{\dfrac{m (1 + \theta)}{\gamma}}.
\end{gather*}
\end{enumerate}
\end{proof}
\par Now let us discuss a few chosen algorithms that implement maximum quasi-biclique search.
\par A greedy algorithm for searching maximum quasi-bicliques according to Definition \ref{def:delta_bicluque} is discussed in detail by \cite{1_quasi_complexity}. The algorithm uses two parameters: 1) $\delta$ to control the size of the quasi-biclique ($\delta = 1-\gamma$) and 2) $\tau$ to control the smallest possible number of vertices that belong to one of the partitions of a quasi-biclique. Let us denote by $U'$ and $V'$ vertex sets of quasi-biclique of the graph $G(U, V, E)$. At the beginning of the algorithm we set $U' = \emptyset$ and $V' = V$. From the vertex set $U\setminus U'$ we choose such vertex $u$ that its degree is maximum and delete from $V$' all vertices for which $d(v, U') < (1-\delta)\cdot |U'|$. This process continues as long as the size of $U'<\tau$. However, this algorithm can miss possible vertex candidates, thus authors introduce the second step of the algorithm: if there is a vertex $u$ outside of the current vertex set $U'$ such that its degree is maximum in $U\setminus U'$ and $U'\cup\{u\}$  remains a quasi-biclique, then it can be added to $U'$. The same applies to $V'$ as long as there is a vertex to add.

\section{Quasi-biclique searching models}\label{sec:models}
\subsection{Model 1}
\par In this section we will show how to adapt the model {\bf F3} from \cite{exact_mip_veremryev} for searching for maximum quasi-bicliques. Let us consider disjoint sets $U'\cup V'$, $U'\cap V'$ = $\emptyset$ that form a quasi-biclique of a bipartite graph $G=(U,V,E)$. Using similar techniques as in \cite{exact_mip_veremryev}, we introduce the following variables:
\begin{gather*}
 u_i=1 \Leftrightarrow i \in U', \\
 v_j=1 \Leftrightarrow j \in V', \\
 y_{ij} = 1 \Leftrightarrow \exists (i, j) \in E \cap \left( U' \times V' \right) \\
 z_k^{(1)} = 1 \Leftrightarrow |U'| = k, z_k^{(2)} = \vert V'\vert,\\
  \omega_l^{(1)}, \omega_u^{(1)} \mbox{are the lower and upper bounds, respectively, for the vertex set } U',\\
  \omega_l^{(2)}, \omega_u^{(2)} \mbox{are the lower and upper bounds, respectively, for the vertex set } V'.
\end{gather*} 
We can refine the sizes of vertex sets of a quasi-biclique using Proposition~\ref{prop:quasi_biclique_up}. Then we build Model 1:

\stepcounter{model_count}
\textbf{\large Model \arabic{model_count}}
\begin{gather}
 \omega_{\gamma}(G) = \max_{u,v,y,z} \left[\sum_{i \in U}{u_i} + \sum_{j \in V}{v_j} \right], \label{eq:model1_max}\\
 \mbox{under conditions} \sum_{(i,j) \in E}{y_{ij}} \geq \gamma \sum_{n = \omega_l^{(1)}}^{\omega_u^{(1)}}\sum_{m=\omega_l^{(2)}}^{\omega_u^{(2)}}{n\cdot m\cdot z_n^{(1)}\cdot z_m^{(2)}},\label{eq:model1_gamma_ineq}\\
 \forall i \in U, \forall j \in V: y_{ij} \leq u_i, y_{ij} \leq v_j, y_{ij} \geq v_i + v_j - 1, 
  \label{eq:model1_yij_restrict}\\
 \sum_{i \in U}{u_i} = \sum_{n=\omega_l^{(1)}}^{\omega_u^{(1)}}{n z_n^{(1)}}, \sum_{j \in V}{v_j} = \sum_{m=\omega^{(2)}_l}^{\omega_u^{(2)}}{m z_m^{(2)}}, \label{eq:model1_sums_vertex_and_zk}\\
 \sum_{n=\omega_l^{(1)}}^{\omega_u^{(1)}}{z_n^{(1)}}=1, \sum_{m=\omega^{(2)}_l}^{\omega_u^{(2)}}{z_m^{(2)}}=1, \label{eq:model1_sums_zk_is1}\\
 \forall i \in U, \forall j \in V: u_i \in \{0,1\}, v_j \in \{0,1\}, \forall i < j, (i, j) \in E : y_{ij} \in \{0,1\}, \  \label{eq:model1_restrictions_uvy}\\
\forall n \in \{\omega_l^{(1)}: z_n^{(1)} \geq 0 \ ,\ldots,\omega_u^{(1)}\}, \forall m \in \{ \omega_l^{(2)},\ldots,\omega_u^{(2)}\}: z_m^{(2)} \geq 0 . \label{eq:model1_restrictions_zk}
\end{gather}
\par As in the model {\bf F3} we can bound $z_k^{(1)}$ and $z_k^{(2)}$ and recast them from binary into continuous variables.
\par Suppose, that there exists an optimal solution $\left(u^*,~v^*,~y^*,~\overline{z^{(1)}},~\overline{z^{(2)}}\right)$ of Model 1, where vectors $\overline{z^{(1)}}$ and $\overline{z^{(2)}}$ are not binary \linebreak $\left(\overline{z_n^{(1)}} \geq 0, \ \overline{z_n^{(2)}} \geq 0\right)$. Let $\widehat{z^{(1)}}$ be a binary vector with $\widehat{z^{(1)}_k} = 1 \Leftrightarrow \vert U' \vert = k$ and  $\widehat{z^{(1)}_k} = 0$ otherwise, where $k \in \{\omega_l^{(1)}, \ldots, \omega_u^{(1)}\}$; analogously, vector $\widehat{z^{(2)}}$: $\displaystyle \widehat{z^{(2)}_k} = 1 \Leftrightarrow \vert V' \vert = k$  and $0$  otherwise. Hence, it is obvious, that vectors $\widehat{z^{(1)}}$ and $\widehat{z^{(2)}}$ satisfy constraint~\ref{eq:model1_sums_zk_is1}. Constraints \ref{eq:model1_gamma_ineq} and \ref{eq:model1_sums_vertex_and_zk} can be rewritten as follows:
\begin{gather*}
\sum_{i \in U}{u^*_i} = \sum_{n=\omega_l^{(1)}}^{\omega_u^{(1)}}{n \widehat{z_n^{(1)}}}, \sum_{j \in V}{v^*_j} = \sum_{m=\omega^{(2)}_l}^{\omega_u^{(2)}}{m \widehat{z_m^{(2)}}} \mbox{ (by definition),}\\
\sum_{(i,j) \in E}{y^*_{ij}} \geq \gamma \sum_{n = \omega_l^{(1)}}^{\omega_u^{(1)}}\sum_{m=\omega_l^{(2)}}^{\omega_u^{(2)}}{n\cdot m\cdot \overline{z_n^{(1)}} \cdot \overline{z_m^{(2)}}} = \\
= \gamma \left(\sum_{n = \omega_l^{(1)}}^{\omega_u^{(1)}}{n \cdot \overline{z_n^{(1)}} }\right) \left(\sum_{m=\omega_l^{(2)}}^{\omega_u^{(2)}}{ m \cdot \overline{z_m^{(2)}}}\right) = \gamma \left(\sum_{i \in U}{u^*_i} \right) \left(\sum_{m=\omega_l^{(2)}}^{\omega_u^{(2)}}{\sum_{j \in V}{v^*_j}}\right) \\
= \gamma \left(\sum_{n = \omega_l^{(1)}}^{\omega_u^{(1)}}{n \cdot \widehat{z_n^{(1)}} }\right) \left(\sum_{m=\omega_l^{(2)}}^{\omega_u^{(2)}}{ m \cdot \widehat{z_m^{(2)}}}\right).
\end{gather*}
\par This means that $\left(u^*, v^*, y^*,\widehat{z^{(1)}}, \widehat{z^{(2)}}\right)$ is also an optimal solution of the problem and usage of continuous variables $z^{(1)}_n$ and $z^{(2)}_m$ in Model 1 is proved.
\par In the worst case, when $\omega_l^{(1)} = \omega_l^{(2)} = 1$, $\omega_u^{(1)} = \vert U \vert$, $\omega_u^{(2)} = \vert V \vert$, the model has $\vert U \vert + \vert V \vert$ binary variables and $\vert E \vert + \vert U \vert + \vert V \vert$ continuous. 
\begin{remark} In Model 1, condition \ref{eq:model1_gamma_ineq} is not linear, so we can linearise it. Let us introduce a new variable $z_{n,m} = z_n^{(1)} \cdot z_m^{(2)}$. Then left side of the inequality \ref{eq:model1_gamma_ineq} is:
\begin{equation*}
\sum_{n=\omega_l^{(1)}}^{\omega_u^{(1)}} \sum_{m=\omega_l^{(2)}}^{\omega_u^{(2)}} (n\cdot m) \cdot z_{n,m}.
\end{equation*}
Conditions \ref{eq:model1_sums_vertex_and_zk} are changed as follows:
\begin{equation*}
\sum_{i \in U}{u_i} = \sum_{n=\omega_l^{(1)}}^{\omega_u^{(1)}} \sum_{m=\omega_l^{(2)}}^{\omega_u^{(2)}}{n z_{n,m}}, \sum_{j \in V}{v_j} = \sum_{n=\omega_l^{(1)}}^{\omega_u^{(1)}} \sum_{m=\omega_l^{(2)}}^{\omega_u^{(2)}}{m z_{n,m}},
\end{equation*}
where $c_{n,m}^{(1)} = n$ and $c_{n,m}^{(2)} = m$.
\par Using this substitution for variables $z_n^{(1)}$ and $z_m^{(2)}$, the model becomes a linear integer programming model. In the worst-case scenario,  for dense graph there are $\vert U \vert + \vert V \vert$ binary variables and $\vert E \vert + \vert U \vert \cdot \vert V \vert$ continuous variables to be optimized.
\end{remark}

\subsection{Model 2}
\par Let us look at different maximizing criteria for related Mixed Integer Programming models. In papers \citep{mirkin_tribox,Mirkin:2011} dedicated to triclustering generation, $K = (G, M, B, I)$ is a triadic context with $G$, the set of objects, $M$, the set of attributes, $B$, the set of conditions,  and $I \subseteq G \times M \times B$, the ternary relation. The proposed triclustering algorithm searches for clusters that maximize the following criteria:
\begin{equation}\label{eq:mirkin_criterion_3d}
	f_3(T) = \rho^2(T) |X| |Y| |Z|.
\end{equation}
\par By narrowing this criteria for binary contexts, it is possible to obtain another maximising criteria for Model 7 {\bf GF3($f$)} from \cite{exact_mip_veremryev}, p.191.
\par For a bipartite graph $G=(U, V, E)$ and its induced subgraph $G[C_1, C_2]$, function $f$ is maximized over the density and size of biclique.
\begin{equation}\label{eq:mirkin_criterion_2d}
f(C_1, C_2) = \rho^2(G[C_1, C_2]) \cdot|C_1|\cdot|C_2| = \dfrac{\left( \lvert \{ (i,j): i\in C_1, j \in C_2, (i,j) \in E\}\lvert\right)^2}{|C_1|\cdot|C_2|}.
\end{equation}
Using variables definitions from the previous model we can rewrite function $f$:
\begin{equation*} 
	f (C) = \dfrac{\left( \sum_{(i,j) \in E}{y_{ij}} \right)^2}{\left( \sum_{i \in U}{u_i} \right) \cdot \left(\sum_{j \in V}{v_j}\right)}
\end{equation*}
\par Since function $f$ is multiplicative, the direct way to transform it to an additive function is logarithmisation:
\begin{gather} 
	f_{log} (C) = 2\cdot \log{\lvert \{ (i,j): i\in C_1, j \in C_2, (i,j) \in E\}\lvert} - \log{\lvert C_1 \lvert} - \log{\lvert C_2 \lvert} =  \notag\\
	\label{eq:mirkin_criterion_2d_loglin} 2\cdot \log{\left(\sum_{(i,j) \in E}{y_{ij}}\right)} - \log{\left(\sum_{i \in U}{u_i}\right)} - \log{\left(\sum_{j \in V}{v_j}\right)}.
\end{gather}
As in Model 1,
\begin{equation*}
\sum_{i \in U}{u_i} = \sum_{n=\omega_l^{(1)}}^{\omega_u^{(1)}}{n z_n^{(1)}}, \sum_{j \in V}{v_j} = \sum_{m=\omega^{(2)}_l}^{\omega_u^{(2)}}{m z_m^{(2)}}.
\end{equation*}
\par Now we introduce a new variable: $w_k = 1 \Leftrightarrow \lvert \{ (i,j): i\in C_1, j \in C_2, (i,j) \in E\}\lvert=k$, then $\sum_{(i,j) \in E}{y_{ij}} = \sum_{(i,j) \in E}{k w_k}$.
\begin{gather}
	f_{log} (C) = 2\cdot \log{\left(\sum_{(i,j) \in E}{k w_k}\right)} - \log{\left(\sum_{n=\omega_l^{(1)}}^{\omega_u^{(1)}}{n z_n^{(1)}}\right)} - \log{\left(\sum_{m=\omega^{(2)}_l}^{\omega_u^{(2)}}{m z_m^{(2)}}\right)} = \notag\\= 2\cdot \sum_{(i,j) \in E}{\log(k) w_k} - \sum_{n=\omega_l^{(1)}}^{\omega_u^{(1)}}{\log{(n)} z_n^{(1)}} - \sum_{m=\omega^{(2)}_l}^{\omega_u^{(2)}}{\log{(m)} z_m^{(2)}}.\label{eq:mirkin_criterion_2d_lin}
\end{gather}
\par Obviously, that equality $\displaystyle \log{\left(\sum_{(i,j) \in E}{k w_k}\right)} = \sum_{(i,j) \in E}{\log(k) w_k}$ because $w_k$ is binary variable and $\displaystyle \sum_{(i,j) \in E}{w_k} = 1$. Thus there exists a unique number $k^*$ such that $w_{k^*}=1$. It follows, that  $\displaystyle \log{\left(\sum_{(i,j) \in E}{k w_k}\right)} = \log{(k^*)} = w_{k^*} \cdot \log{(k^*)} = \sum_{(i,j) \in E}{\log(k) w_k}$. A similar statement is true for $\displaystyle \log{\left(\sum_{n=\omega_l^{(1)}}^{\omega_u^{(1)}}{n z_n^{(1)}}\right)}$ and\  $\displaystyle \log{\left(\sum_{m=\omega^{(2)}_l}^{\omega_u^{(2)}}{m z_m^{(2)}}\right)}$. 
\par Without extra conditions on the sizes of vertex sets of a quasi-biclique and its minimum number of edges, the model has $2 \cdot \left( \vert U \vert + \vert V \vert \right) + 2 \cdot \vert E \vert$ variables.\\
\\
\stepcounter{model_count}
\textbf{\large Model \arabic{model_count}}
\begin{gather*}
2 \sum_{k=1}^{|E|}{\log{(k)}\cdot w_{k}} - \sum_{n=1}^{|U|}{\log{(n)} z_n^{(1)}} - \sum_{m=1}^{|V|}{\log{(m)} z_m^{(2)}} \xrightarrow[w,z^{(1)},z^{(2)}]{} max,  \label{eq:model_2_criterion}\\
\mbox{under conditions } \sum_{k=1}^{\vert E \vert}{w_k} \geq \gamma \sum_{n = \omega_l^{(1)}}^{\omega_u^{(1)}} \sum_{m=\omega_l^{(2)}}^{\omega_u^{(2)}}{n\cdot m\cdot z_n^{(1)}\cdot z_m^{(2)}},\label{eq:model_2_gamma_ineq}\\
\sum_{(i,j) \in E}{y_{ij}} = \sum_{k=1}^{|E|}{k\cdot w_{k}} ,\label{eq:model_2_sum_y_wk} \\
\sum_{i \in U}{u_i} = \sum_{n = \omega_l^{(1)}}^{\omega_u^{(1)}}{n z_n^{(1)}}, \sum_{j \in V}{v_j} = \sum_{m=\omega_l^{(2)}}^{\omega_u^{(2)}}{m z_m^{(2)}}, \label{eq:model_2_sums_uv_zk}\\
 \sum_{k=1}^{|E|}{w_k}=1, \sum_{n=1}^{|U|}{z_n^{(1)}}=1, \sum_{m=1}{z_m^{(2)}=1},\label{eq:model_2_sums_is1}\\
\forall i \in U: u_i \in \{0,1\} \, \forall j \in V: v_j \in \{0,1\} \, \label{eq:model_2_restrictions_uv}\\
\forall i < j, (i, j) \in E: y_{ij} \in \{0,1\} \ ,\ \forall k \in \{1,\ldots,|E|\}:  w_k \in \{0,1\} , \label{eq:model_2_restrictions_yw}\\
\forall n \in \{\omega_l^{(1)},\ldots,\omega_u^{(2)}\}: z_n^{(1)} \geq 0 , \ \forall m \in \{\omega_l^{(2)},\ldots,\omega_u^{(2)}\}: z_m^{(2)} \geq 0  \label{eq:model_2_restrictions_zk}. 
\end{gather*}
\begin{remark}
In order to simplify the model. we can add extra constraints for variables $w_k,\ k \in \{1,\ldots,\vert E \vert\}$. Let $k$ be a possible number of edges in a quasi-biclique, then:
\begin{enumerate}
\item $k \leq \omega_u^{(1)} \cdot \omega_u^{(2)}$.

\item If $\gamma \cdot \omega_l^{(1)} \cdot \omega_l^{(2)} \leq \vert E \vert \Rightarrow k \geq \gamma \cdot \omega_l^{(1)} \cdot \omega_l^{(2)}$.

\item Let us consider $U'$ such that $\vert U' \vert = \omega_l^{(1)}$ and $\displaystyle \forall u \in U'$ $deg(u) \leq \min_{x\in U\setminus U'}\{deg(x)\}$. That is $U'$ is a subset of $U$ with the minimum possible size and with all smallest degree vertices w.r.t. $U$. Then $k \geq \displaystyle\gamma \sum_{u \in U'}{deg(u)}$.

\item Similarly, for $V' \subseteq V$: $\vert V' \vert = \omega_l^{(2)}$ and $\displaystyle \forall v \in V' deg(v) \leq \min_{x\in V\setminus V'}\{deg(x)\}$, then $k \geq \displaystyle\gamma \sum_{v \in V'}{deg(v)}$.
\end{enumerate}

\end{remark}

\section{Datasets}\label{sec:data}
Datasets for testing the performance of the algorithms are mainly taken from \citep{Borgatti:2014,Batagelj:2014}.
\begin{enumerate}

\item Southern Women: $|U|=18, |V|=14$, $|E|=89$ edges, a classic ethnographic dataset with a bipartite graph of 18 women, which met in a series of 14 informal social events~\citep{Freeman:2003}.
\item Divorce in the US: $|U|=9, |V|=50$ vertices, and $|E|=225$ edges. This graph describes the particular causes of divorce in the United States.
\item Dutch Elite: $|U|=3810, |V|=937$ vertices,  and $|E|=5221$ edges. This graph describes the connections between people and the most important for the Netherlands government administrative authorities.
\item Dutch Elite (TOP-200): $|U|=200, |V|=395$ vertices,  and $|E|=877$ edges. The list of people in the first partition of the graph consists of the most influential persons regarding their membership in administrative authorities.
\item Movie-Lens (ml-latest-small): $|V|=99125, |V|=50$ vertices,  and $|E|=20340$ edges, a bipartite graph of ``movie-genre'' relation from Movie Lens project~\citep{Harper:2015}.
\end{enumerate}

\section{Experimental Verification}\label{sec:experiments}
\subsection{Implementation description}
The greedy algorithm of searching for maximal $\gamma$-quasi-biclique in a bipartite graph was implemented in Python 2.7. The MIP models were implemented with the optimization package CPLEX, created by IBM. 
All computations were carried out on a laptop with macOS operating system, 2.7 GHz Intel Core i5 processor, and RAM 8 GB 1867 MHz.

The search for solutions in the CPLEX package was performed by means of the branch-and-cut method, which is similar to the  branch-and-bound algorithmic approach. The method uses a search tree, where each node represents a subproblem that needs to be solved and possibly analysed further.

The \textsc{branch} procedure creates two new nodes from the active parent node. Generally, at this point, the boundaries of one variable are applied and stored for the current node and all its child nodes. In its turn, the \textsc{cut} procedure adds a new constraint to the model. As a result of any cut, the solution space for the subproblems, which are presented in the nodes, is reduced, and the number of branches needed to process decreases. CPLEX processes active nodes in the tree until no more active nodes are available or a certain limit is reached \footnote{CPLEX user manual: https://www.ibm.com/support/knowledgecenter/SSSA5P\_ 12.8.0/ilog.odms.cplex.help/CPLEX/homepages/usrmancplex.html }.

The standard solution with the CPLEX software package assumes only one of the optimal solutions as the answer. However, in CPLEX it is possible to obtain a set of optimal solutions using the \textit{solution pool} method, which allows one to find and store several solutions of MIP models.

The generation of multiple solutions works in two steps. The first step is identical to the usual solution search using the CPLEX software package. At this step, the algorithm finds the only optimal solution of the integer programming problem. It also saves nodes in the search tree that could potentially be useful; for example, if not all the variable constraints are taken into account or if all the nodes contain a suitable value, but the target function is not optimal.

In the second step, using previously calculated and stored  information in the first stage, several solutions are generated, and the tree is traversed again, in particular within the branches rooted from the additional nodes stored in the first stage.

\subsection{Illustrative examples}
On a toy example of a graph with 12 vertices, we consider the search results for maximal $\gamma$-quasi-bicliques, $\gamma=0.8$, using Models 1 and 2, respectively (Fig.~\ref{img:models_simple_example}).

\begin{figure}[h!]
 \begin{center}
 \begin{minipage}{\linewidth}
    \center \includegraphics[width = \textwidth]{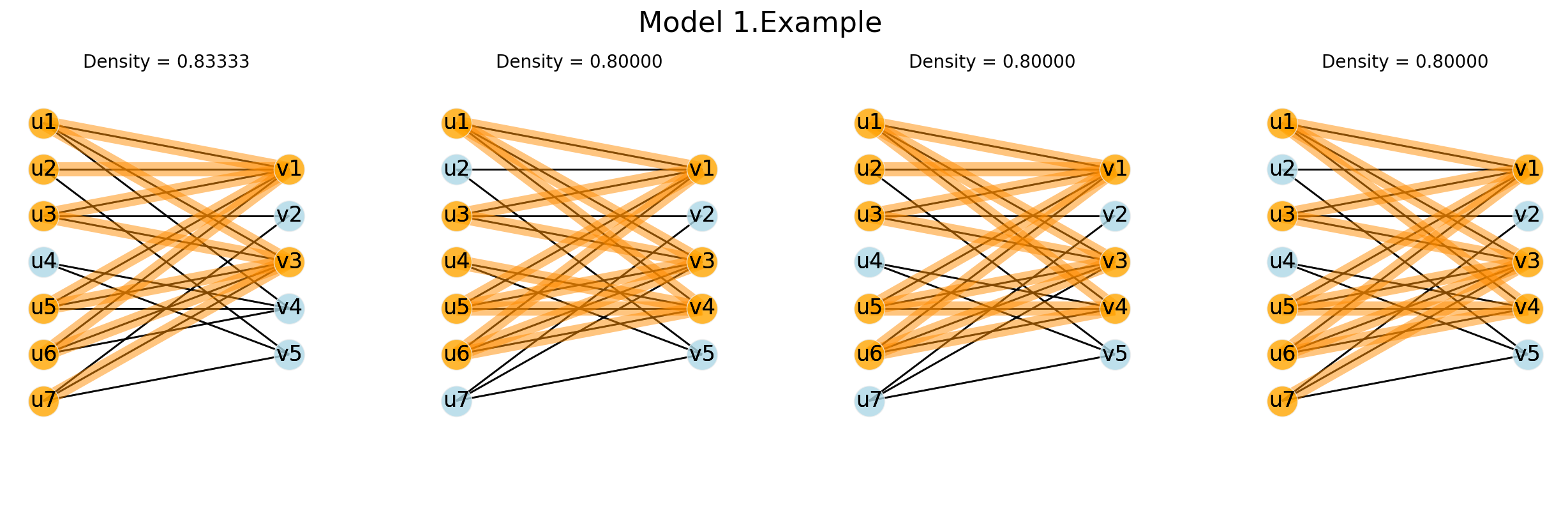}
    \end{minipage}
 \vfill
 \begin{minipage}{\linewidth}
  \center \includegraphics[width = \textwidth]{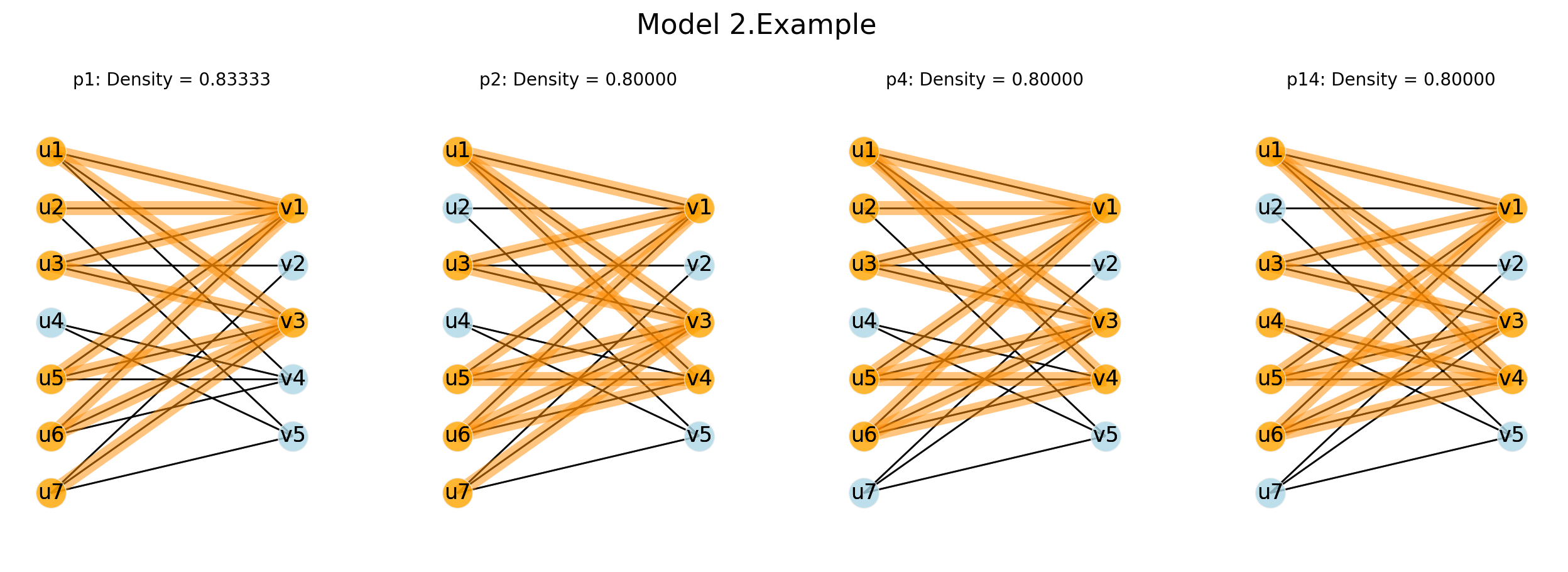}
 \end{minipage}
  \caption{The results of search for quasi-bicliques using Models 1 and 2.}
   \label{img:models_simple_example}
 \end{center}
\end{figure}

The results for both models are the same (w.r.t. to the solutions output order). Even for this small-sized problem, the time is tangible: the computations with Model 1 took 2.16 s, and for Model 2, it was 2.94 s. A comparison of the executed models and the greedy algorithm in terms of computational time is given below for the selected bipartite graphs.

\subsection{Comparison of algorithms }
The algorithm of searching for the maximal quasi-biclique using the CPLEX software package was implemented for Models 1 and 2 (see Section~\ref{sec:models}) and compared with the \textsc{greedy} algorithm from \citep{1_quasi_complexity} (let us denote it as Greedy Algorithm). 

There are no comparison results presented for the model {\bf F3} \citep{exact_mip_veremryev}: despite its fast work, the algorithm based on this model chose quasi-bicliques of very small size and maximum density (i.e. bicliques). This phenomenon is rather expected since the model {\bf F3} implies a completely different function of the density of the subgraph. Therefore, the comparison, in this case, is irrelevant. The description of Complete QB in~ \citep{3_financial_ratios} lacks of important implementation details.

The weakness of the constructed MIP models was identified during the finding solution. Since the problem of enumerating all maximal quasi-bicliques in practice requires considerable time, the software package can discard some solutions, if it has found quite a few optimal ones already. First of all, the search is carried out among unbalanced quasi-bicliques (no constraints on the approximately equal size of the quasi-clique partitions have been given). For large graphs, this means that the number of vertices in one of the parts of the found optimal solution may exceed the number of vertices in the second part by hundreds of times or more.

This issue can be addressed in two ways. Firstly, one can set roughly equal limits on the size of the partitions. Secondly, it is possible to adapt the model for finding an almost balanced quasi-bicliques, that means that sizes of partitions of a quasi-biclique differ by $\theta$. To do this, the following conditions should be added to Model 1 or Model 2:

\begin{equation}
\sum_{n=\omega_l^{(1)}}^{\omega_u^{(1)}}{z_n^{(1)}} \geq (1~-~\theta)~\sum_{m=\omega^{(2)}_l}^{\omega_u^{(2)}}{z_m^{(2)}}, \label{eq:balanced_condition_lower}
\end{equation}
\begin{equation}
\sum_{n=\omega_l^{(1)}}^{\omega_u^{(1)}}{z_n^{(1)}} \leq (1~+~\theta)~\sum_{m=\omega^{(2)}_l}^{\omega_u^{(2)}}{z_m^{(2)}}\label{eq:balanced_condition_upper}.
\end{equation}

Models with additional conditions \ref{eq:balanced_condition_lower} and \ref{eq:balanced_condition_upper} have not been tested.

It has also been noted that small-sized quasi-bicliques can be useless in practice, but their recalculation is costly. Therefore, for each data set, we can establish minimum bounds on the size of a quasi-biclique (of the order of the smallest vertex degree with respect to the graph partitions).

The results the algorithms execution are presented in Table~\ref{tbl:results_06} for $\gamma=0.6$, Table~\ref{tbl:results_07} for $\gamma=0.7$  and in Table~\ref{tbl:results_08} for $\gamma=0.8$\footnote{The size column in Table~\ref{tbl:results_08} shown as the result of summation $|U'|$ and $|V'|$}. For each algorithm its main parameters are indicated: the algorithm running time (time), the number of found maximum quasi-bicliques (count) and the maximum size of the found solution.

 \begin{table}[h!]
 \caption{Results of maximum $\gamma$-quasi-biclique search. Parameters: $\gamma = 0.6$.}
 \label{tbl:results_06}
 \begin{center}
 \begin{tabular}{|c||*{3}{c|}|*{3}{c|}|*{3}{c|}}
  \hline
  \multirow{2}*{Data} & \multicolumn{3}{c||}{\textbf{\small Model 1}} & \multicolumn{3}{c||}{\textbf{\small Model 2}} &  \multicolumn{3}{c|}{\textbf{\small Greedy Algorithm} }  \\ \cline{2-10}
   & \small{time} & \small{count} & \small{size} &\small{time} & \small{count} & \small{size} & \small{time} & \small{count} & \small{size}  \\ \hline 
   \small{Southern} & \multirow{2}*{678 ms}& \multirow{2}*{4} & \multirow{2}*{(18,4)} & \multirow{2}*{801 ms} & \multirow{2}*{2} & \multirow{2}*{(18,4)} & \multirow{2}*{234 ms} & \multirow{2}*{4} & \multirow{2}*{(17, 5)}  \\ 
   \small{Women} & & & & & & & & & \\ \hline 
   \small{Divorse} & \multirow{2}*{1.23 s}& \multirow{2}*{1} & \multirow{2}*{(4,50)} & \multirow{2}*{3.38 s} & \multirow{2}*{1} & \multirow{2}*{(4,50)} & \multirow{2}*{360 ms} & \multirow{2}*{1} & \multirow{2}*{(2, 46)}  \\ 
   \small{in US} & & & & & & & & & \\ \hline 
   \small{DutchElite} & \multirow{2}*{7602 s} & \multirow{2}*{2} & \multirow{2}*{(26,1)} & \multirow{2}*{181 s} & \multirow{2}*{1} & \multirow{2}*{(11,3)} & \multirow{2}*{3 s} & \multirow{2}*{1} & \multirow{2}*{(10,3)}  \\
   \small{(top200)} & & & & & & & & & \\ \hline  
   \small{DutchElite} & - & - & - & 6968 s & 1 & (45,2) & 1954 s& 1 & (40,2)  \\ \hline 
   \small{Movie-Lens} & \multirow{2}*{28068 s} & \multirow{2}*{2} & \multirow{2}*{(692,2)} & \multirow{2}*{13851 s} & \multirow{2}*{5} & \multirow{2}*{(900,3)} & \multirow{2}*{5976 s} & \multirow{2}*{2} & \multirow{2}*{(754,2)} \\ 
   \small{(small)} & & & & & & & & & \\ \hline  
  \end{tabular}
 \end{center} 
\end{table}

 \begin{table}[h!]
 \caption{The results of maximum $\gamma$-quasi-biclique search for $\gamma = 0.7$.}

 \label{tbl:results_07}
 \begin{center}
 \begin{tabular}{|c||*{3}{c|}|*{3}{c|}|*{3}{c|}}
  \hline
  \multirow{2}*{Data} & \multicolumn{3}{c||}{\textbf{Model 1}} & \multicolumn{3}{c||}{\textbf{Model 2}} &  \multicolumn{3}{c|}{\textbf{\small Greedy' algorithm} }  \\ \cline{2-10}
    & \small{time} & \small{count} & \small{size} &\small{time} & \small{count} & \small{size} & \small{time} & \small{count} & \small{size}  \\ \hline 
   \small{Southern} & \multirow{2}*{1.29 s}& \multirow{2}*{1} & \multirow{2}*{(16,3)} & \multirow{2}*{1.11 s} & \multirow{2}*{1} & \multirow{2}*{(10,6)} & \multirow{2}*{309 ms} & \multirow{2}*{1} & \multirow{2}*{(16, 2)}  \\ 
   \small{Women} & & & & & & & & & \\ \hline 
   \small{Divorse} & \multirow{2}*{1.56 s} & \multirow{2}*{1} & \multirow{2}*{(2,45)} & \multirow{2}*{2.66 s} & \multirow{2}*{3} & \multirow{2}*{(5,36)} & \multirow{2}*{320 ms} & \multirow{2}*{1} & \multirow{2}*{(2,28)} \\
   \small{in US} & & & & & & & & & \\ \hline 
   \small{DutchElite} & \multirow{2}*{8497 s} & \multirow{2}*{1} & \multirow{2}*{(23,1)} & \multirow{2}*{1668 s} & \multirow{2}*{3} & \multirow{2}*{(10,3)} & \multirow{2}*{1.63 s} & \multirow{2}*{1} & \multirow{2}*{(10,3)}  \\
   \small{(top200)} & & & & & & & & & \\ \hline 
   \small{DutchElite} & - & - & - & 6166 s & 1 & (20,2) & 1511 s & 1 & (20,1)  \\ \hline 
   \small{Movie-Lens} & \multirow{2}*{-} & \multirow{2}*{-} & \multirow{2}*{-} & \multirow{2}*{10719 s} & \multirow{2}*{6} & \multirow{2}*{(800, 3)} & \multirow{2}*{-} & \multirow{2}*{-} & \multirow{2}*{-} \\ 
   \small{(small)} & & & & & & & & & \\ \hline 
  \end{tabular}
 \end{center} 
\end{table}

Two found $\gamma$-quasi bicliques for the dataset divorce in the US are shown in Fig.~\ref{MIP:divorce}.

\begin{figure}[h!]
 \begin{center}
 \begin{minipage}{0.49\linewidth}
 \begin{center}
    \includegraphics[width = 1\textwidth]{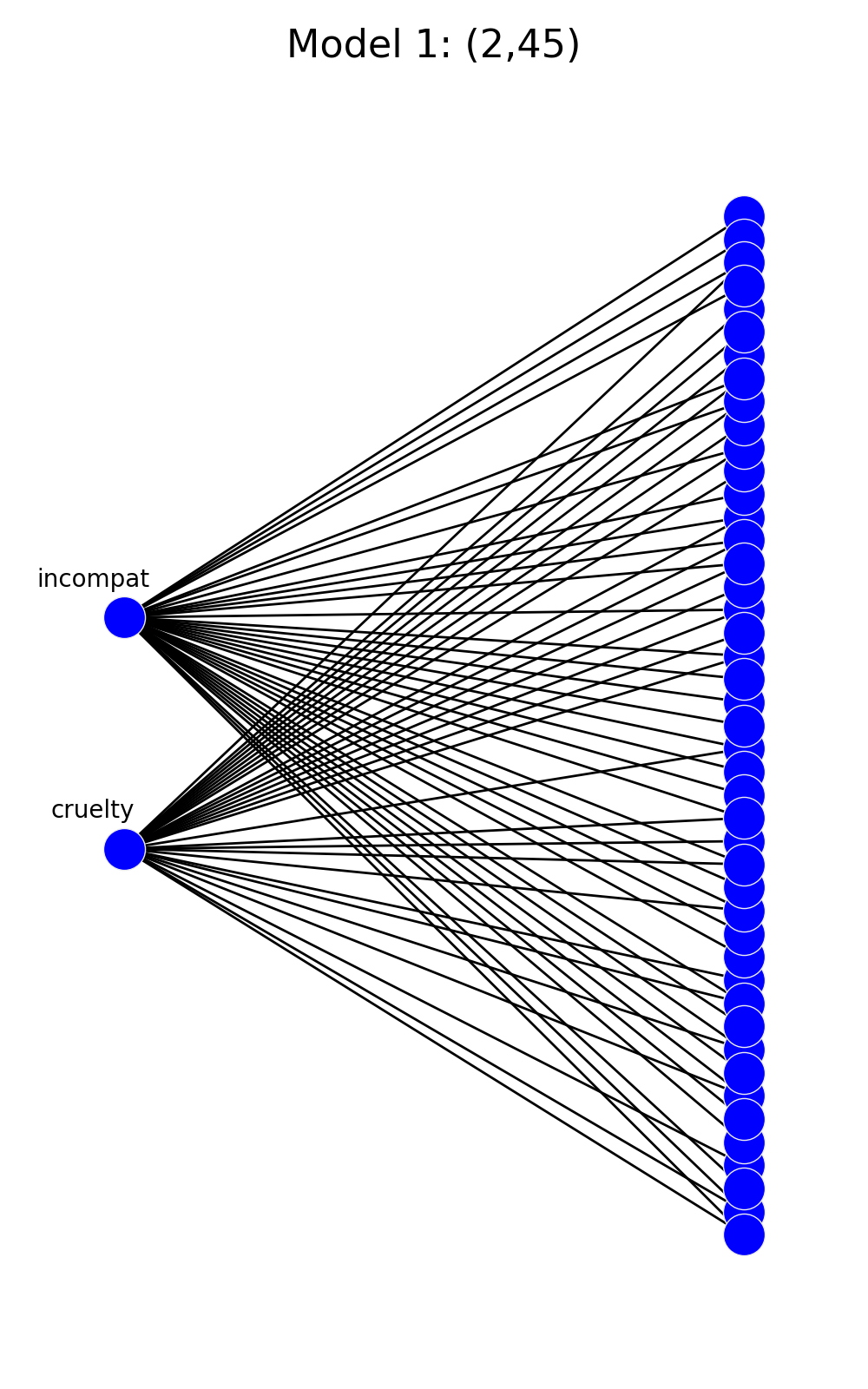}
\label{img:model1_divorce}
   \end{center}    
  \end{minipage}  
   \hfill
  \begin{minipage}[h]{0.49\linewidth}
 	 \includegraphics[width = 1\textwidth]{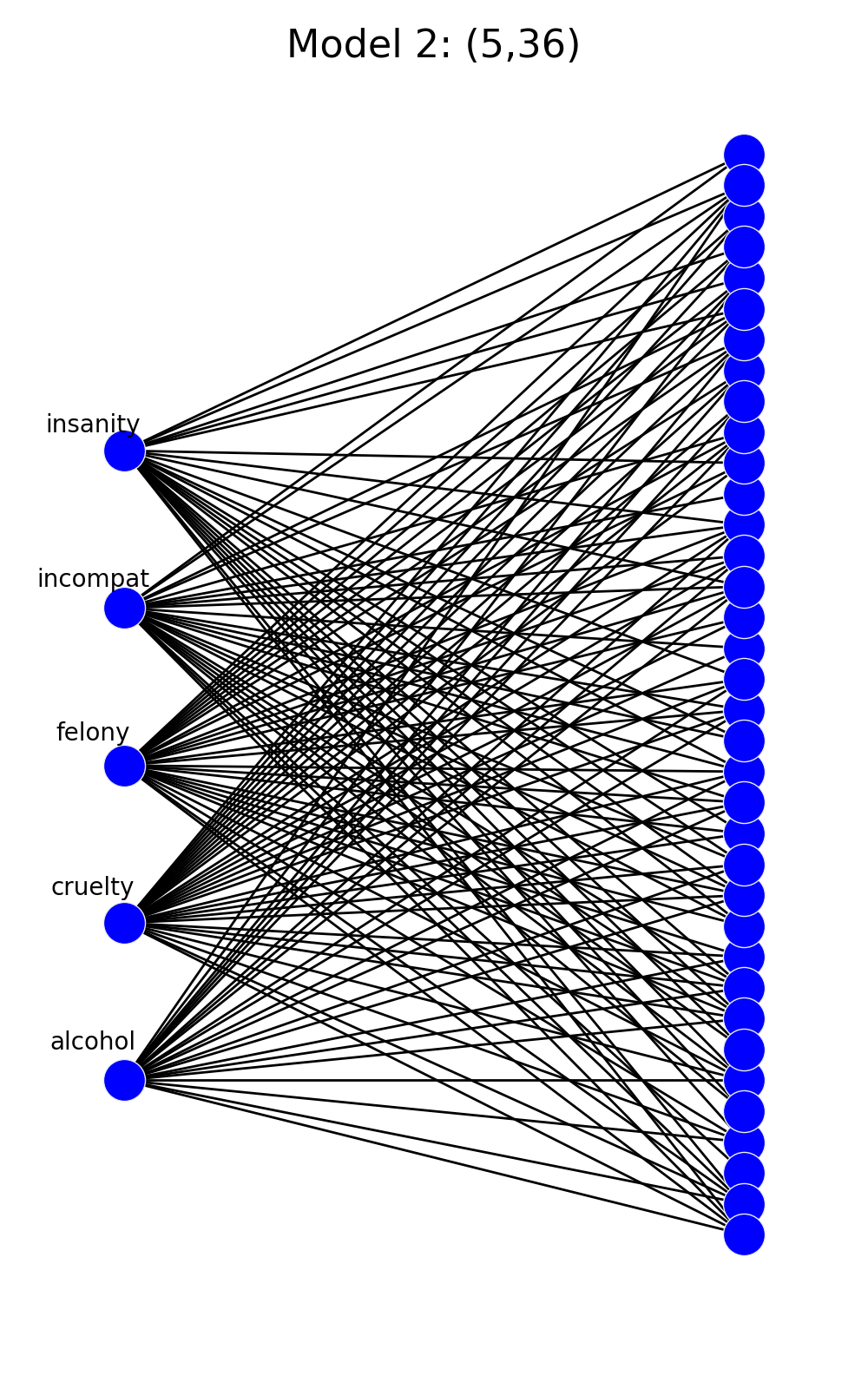}
\label{img:model2_divorce}    
  \end{minipage}  
 \end{center}
 \caption{Quasi-bicliques obtained by the studied MIP models for the dataset Divorce in US with $\gamma=0.7$.}\label{MIP:divorce}
\end{figure}

 \begin{table}[h!]
 \caption{The results of maximum $\gamma$-quasi-biclique search for  $\gamma = 0.8$.}

 \label{tbl:results_08}
 \begin{center}
  \begin{tabular}{|c||*{3}{c|}|*{3}{c|}|*{3}{c|}}
  \hline
  \multirow{2}*{Data} & \multicolumn{3}{c||}{\textbf{Model 1}} & \multicolumn{3}{c||}{\textbf{Model  2}} &  \multicolumn{3}{c|}{\textbf{Greedy algorithm} }\\ \cline{2-10} 
    & \small{time} & \small{count} & \small{size} &\small{time} & \small{count} & \small{size} & \small{time} & \small{count} & \small{size}   \\ \hline 
   \small{Divorce in US} & 8.53 s & 1 & 38 & 1.7 s & 2 & 33 & 313 ms & 1 & 25  \\ \hline 
   \small{DutchElite (top200)} & - & - & - & 4834 s & 2 & 13 &  2.5 s & 1 & 13  \\ \hline 
   \small{DutchElite} & - & - & - & 7129 s & 1 & 47 & 1718 & 1 & 21  \\ \hline 
   \small{Movie-Lens} (small) & - & - & - & 9046 s & 2 & 445 & - & - & - \\ \hline 
  \end{tabular}
 \end{center} 
\end{table}

Dashes ("-") in the following tables mean that the algorithm worked 10 hours and did not find a solution. If one of the partitions of the maximal quasi-biclique has a unit size, this is marked in the table as $(U', V')$, where $U'$ and $V'$  are the sizes of the partitions.

\section{Results and conclusions}\label{sec:results_concl}
One can note, that mixed linear programming models work an order of magnitude slower than the greedy algorithm by \cite{1_quasi_complexity}, but they find more quasi-cliques and generally each of them has a larger size.

For small graphs, the time for finding the solution by the considered models is acceptable. Model 1 contains a fewer number of variables that must be optimized, but its maximisation criterion is costly for large graphs. Thus, on large-sized graphs Model 1 works too long (more than 10 hours), especially for high $\gamma$ density thresholds. The dependence of the speed and quality of processing on $\gamma$ is also apparent for two other algorithms: for high thresholds on density, those methods work longer since the number of possible optimal solutions to the problem is reduced.
Model 2 on similar graphs showed better results, but the processing time is still quite large. For $Dutch Elite$ data with a large number of vertices and a small number of edges, MIP-based algorithms work much longer than on more dense graphs.

If we consider the results, not in terms of speed, but terms of quality, then Model 2 was the best one. This model produced more unique and larger quasi-bicliques than other algorithms. 

The following ways of future work seems to be relevant: 1) further improvements of the proposed models by establishing tighter bounds for different constraints and using optimization tricks; 2) exploration of new optimization criteria; 3) comparison of different MIP solvers with the state-of-the-art approaches of searching for quasi-bicliques in a larger set of experiments.

Another interesting avenue for research could be a study on connection between various approximations of formal concepts (fault-tolerant concepts~\citep{Besson:2006} and object-attribute biclusters~\citep{Ignatov:2012,Ignatov:2017}), Boolean matrix factorization~\citep{Miettinen:2013,Belohlavek:2019} and quasi-bicliques.

\begin{acknowledgement}
The work of  Dmitry I. Ignatov  shown  in  all the sections, except 5 and 6, has been supported by the Russian Science Foundation grant no. 17-11-01276 and performed at St. Petersburg Department of Steklov Mathematical Institute of Russian Academy of Sciences, Russia. The authors would like to thank Boris Mirkin, Vladimir Kalyagin, Panos Pardalos, and Oleg Prokopyev for their piece of advice and inspirational discussions. Last but not least we are thankful to anonymous reviewers for their useful  feedback.
\end{acknowledgement}

\bibliographystyle{spbasic} 
\bibliography{mipqbic} 

\begin{thebibliography}{19}
\providecommand{\natexlab}[1]{#1}
\providecommand{\url}[1]{{#1}}
\providecommand{\urlprefix}{URL }
\expandafter\ifx\csname urlstyle\endcsname\relax
  \providecommand{\doi}[1]{DOI~\discretionary{}{}{}#1}\else
  \providecommand{\doi}{DOI~\discretionary{}{}{}\begingroup
  \urlstyle{rm}\Url}\fi
\providecommand{\eprint}[2][]{\url{#2}}

\bibitem[{Abello et~al.(2002)Abello, Resende, and
  Sudarsky}]{5_massive_quasi_cliques}
Abello J, Resende MGC, Sudarsky S (2002) Massive quasi-clique detection. In:
  Rajsbaum S (ed) LATIN 2002: Theoretical Informatics, Springer Berlin
  Heidelberg, Berlin, Heidelberg, pp 598--612

\bibitem[{Batagelj and Mrvar(2014)}]{Batagelj:2014}
Batagelj V, Mrvar A (2014) Pajek. In: Encyclopedia of Social Network Analysis
  and Mining, pp 1245--1256, \doi{10.1007/978-1-4614-6170-8\_310},
  \urlprefix\url{https://doi.org/10.1007/978-1-4614-6170-8\_310}

\bibitem[{Belohl{\'{a}}vek et~al.(2019)Belohl{\'{a}}vek, Outrata, and
  Trnecka}]{Belohlavek:2019}
Belohl{\'{a}}vek R, Outrata J, Trnecka M (2019) Factorizing boolean matrices
  using formal concepts and iterative usage of essential entries. Inf Sci
  489:37--49, \doi{10.1016/j.ins.2019.03.001},
  \urlprefix\url{https://doi.org/10.1016/j.ins.2019.03.001}

\bibitem[{Besson et~al.(2006)Besson, Robardet, and Boulicaut}]{Besson:2006}
Besson J, Robardet C, Boulicaut J (2006) Mining a new fault-tolerant pattern
  type as an alternative to formal concept discovery. In: Conceptual
  Structures: Inspiration and Application, 14th International Conference on
  Conceptual Structures, {ICCS} 2006, Aalborg, Denmark, July 16-21, 2006,
  Proceedings, pp 144--157, \doi{10.1007/11787181\_11},
  \urlprefix\url{https://doi.org/10.1007/11787181\_11}

\bibitem[{Borgatti et~al.(2014)Borgatti, Everett, and Freeman}]{Borgatti:2014}
Borgatti SP, Everett MG, Freeman LC (2014) {UCINET}. In: Encyclopedia of Social
  Network Analysis and Mining, pp 2261--2267,
  \doi{10.1007/978-1-4614-6170-8\_316},
  \urlprefix\url{https://doi.org/10.1007/978-1-4614-6170-8\_316}

\bibitem[{Freeman(2003)}]{Freeman:2003}
Freeman LC (2003) Finding social groups: A meta-analysis of the southern women
  data. In: Breiger R, Carley K, Pattison P (eds) Dynamic Social Network
  Modeling and Analysis: Workshop Summary and Papers, National Academies Press

\bibitem[{Harper and Konstan(2015)}]{Harper:2015}
Harper FM, Konstan JA (2015) The movielens datasets: History and context. ACM
  Trans Interact Intell Syst 5(4):19:1--19:19, \doi{10.1145/2827872},
  \urlprefix\url{http://doi.acm.org/10.1145/2827872}

\bibitem[{Ignatov et~al.(2012)Ignatov, Kuznetsov, and Poelmans}]{Ignatov:2012}
Ignatov DI, Kuznetsov SO, Poelmans J (2012) Concept-based biclustering for
  internet advertisement. In: 12th {IEEE} International Conference on Data
  Mining Workshops, {ICDM} Workshops, Brussels, Belgium, December 10, 2012, pp
  123--130, \doi{10.1109/ICDMW.2012.100},
  \urlprefix\url{https://doi.org/10.1109/ICDMW.2012.100}

\bibitem[{Ignatov et~al.(2015)Ignatov, Gnatyshak, Kuznetsov, and
  Mirkin}]{mirkin_tribox}
Ignatov DI, Gnatyshak DV, Kuznetsov SO, Mirkin BG (2015) Triadic formal concept
  analysis and triclustering: searching for optimal patterns. Machine Learning
  101(1):271--302, \doi{10.1007/s10994-015-5487-y},
  \urlprefix\url{https://doi.org/10.1007/s10994-015-5487-y}

\bibitem[{Ignatov et~al.(2017)Ignatov, Semenov, Komissarova, and
  Gnatyshak}]{Ignatov:2017}
Ignatov DI, Semenov A, Komissarova D, Gnatyshak DV (2017) Multimodal clustering
  for community detection. In: Formal Concept Analysis of Social Networks, pp
  59--96, \doi{10.1007/978-3-319-64167-6\_4},
  \urlprefix\url{https://doi.org/10.1007/978-3-319-64167-6\_4}

\bibitem[{Liu et~al.(2008{\natexlab{a}})Liu, Liu, and Wang}]{6_protein_protein}
Liu HB, Liu J, Wang L (2008{\natexlab{a}}) Searching maximum quasi-bicliques
  from protein-protein interaction network. Journal of Biomedical Science and
  Engineering 1(03):200

\bibitem[{Liu et~al.(2008{\natexlab{b}})Liu, Li, and Wang}]{1_quasi_complexity}
Liu X, Li J, Wang L (2008{\natexlab{b}}) Quasi-bicliques: Complexity and
  binding pairs. In: Hu X, Wang J (eds) Computing and Combinatorics, Springer
  Berlin Heidelberg, Berlin, Heidelberg, pp 255--264

\bibitem[{Miettinen(2013)}]{Miettinen:2013}
Miettinen P (2013) Fully dynamic quasi-biclique edge covers via boolean matrix
  factorizations. In: Proceedings of the Workshop on Dynamic Networks
  Management and Mining, ACM, New York, NY, USA, DyNetMM '13, pp 17--24,
  \doi{10.1145/2489247.2489250},
  \urlprefix\url{http://doi.acm.org/10.1145/2489247.2489250}

\bibitem[{Mirkin and Kramarenko(2011)}]{Mirkin:2011}
Mirkin BG, Kramarenko AV (2011) Approximate bicluster and tricluster boxes in
  the analysis of binary data. In: Rough Sets, Fuzzy Sets, Data Mining and
  Granular Computing - 13th International Conference, RSFDGrC 2011, Moscow,
  Russia, June 25-27, 2011. Proceedings, pp 248--256,
  \doi{10.1007/978-3-642-21881-1\_40},
  \urlprefix\url{https://doi.org/10.1007/978-3-642-21881-1\_40}

\bibitem[{Pattillo et~al.(2013)Pattillo, Veremyev, Butenko, and
  Boginski}]{8_on_maximum_quasi_clique}
Pattillo J, Veremyev A, Butenko S, Boginski V (2013) On the maximum
  quasi-clique problem. Discrete Applied Mathematics 161(1):244 -- 257,
  \doi{https://doi.org/10.1016/j.dam.2012.07.019},
  \urlprefix\url{http://www.sciencedirect.com/science/article/pii/S0166218X12002843}

\bibitem[{Peeters(2003)}]{Peeters:2003}
Peeters R (2003) The maximum edge biclique problem is np-complete. Discrete
  Applied Mathematics 131(3):651 -- 654,
  \doi{https://doi.org/10.1016/S0166-218X(03)00333-0}

\bibitem[{Sim et~al.(2006)Sim, Li, Gopalkrishnan, and Liu}]{3_financial_ratios}
Sim K, Li J, Gopalkrishnan V, Liu G (2006) Mining maximal quasi-bicliques to
  co-cluster stocks and financial ratios for value investment. In: Sixth
  International Conference on Data Mining (ICDM'06), pp 1059--1063,
  \doi{10.1109/ICDM.2006.111}

\bibitem[{Veremyev et~al.(2016)Veremyev, Prokopyev, Butenko, and
  Pasiliao}]{exact_mip_veremryev}
Veremyev A, Prokopyev OA, Butenko S, Pasiliao EL (2016) Exact mip-based
  approaches for finding maximum quasi-cliques and dense subgraphs. Comp Opt
  and Appl 64(1):177--214, \doi{10.1007/s10589-015-9804-y},
  \urlprefix\url{https://doi.org/10.1007/s10589-015-9804-y}

\bibitem[{Wang(2013)}]{Wang:2013}
Wang L (2013) Near optimal solutions for maximum quasi-bicliques. Journal of
  Combinatorial Optimization 25(3):481--497, \doi{10.1007/s10878-011-9392-4}

\end{thebibliography}

\end{document}